\def\year{2021}\relax
\documentclass[letterpaper]{article} 
\usepackage{aaai21}  
\usepackage{times}  
\usepackage{helvet} 
\usepackage{courier}  
\usepackage[hyphens]{url}  
\usepackage{graphicx} 
\usepackage{natbib}  
\usepackage{caption} 
\title{Double Oracle Algorithm for Computing Equilibria in Continuous Games}
 \author{
     Luk\'a\v{s Adam}\textsuperscript{\rm 1}, Rostislav Hor\v{c}\'ik\textsuperscript{\rm 1}, Tom\'a\v{s} Kasl\textsuperscript{\rm 1}, Tom\'a\v{s} Kroupa 
     \textsuperscript{\rm 1}\\
 }
 \affiliations{
     \textsuperscript{\rm 1}Department of Computer Science, Faculty of Electrical Engineering, Czech Technical University in Prague\\
     \{
       adamluk3, xhorcik, kasltoma, tomas.kroupa\}@fel.cvut.cz

 }

\usepackage{amsmath,amssymb,amsthm}
\usepackage[color=green!40]{todonotes}
\usetikzlibrary{patterns}
\usepackage{mathtools}
\usepackage{xfrac}
\usepackage{comment}
\usepackage{tikz,pgfplots}
\usepackage{tikz-cd}
\usepackage{mathrsfs}
\usepackage{bbm}
\usepackage{environ}
\makeatletter
\providecommand{\env@tikzpicture@save@env}{}
\providecommand{\env@tikzpicture@process}{}

\theoremstyle{definition}

\theoremstyle{plain}
\newtheorem{theorem}{Theorem}[section]

\newtheorem{lemma}{Lemma}[section]
\newtheorem{proposition}{Proposition}[section]

\theoremstyle{definition}

\newtheorem{example}{Example}[section]

\DeclarePairedDelimiterX\set[1]\lbrace\rbrace{#1}

\DeclareMathOperator*{\spt}{spt}

\DeclareMathOperator{\R}{\mathbb{R}}

\newcommand{\calG}{\mathcal{G}}

\newcommand{\eps}{\epsilon}

\newcommand{\cl}{\operatorname{cl}}

\usepackage[section]{algorithm}
\usepackage{algorithmicx}
\usepackage{algpseudocode}

\usepackage{bm}

\pgfplotsset{compat=1.14}

\pgfplotsset{
    lineA/.style      = {blue, very thick},
    lineB/.style      = {black, very thick, dashdotted},
    lineC/.style      = {blue, line width=2mm},
    linetau/.style      = {black, very thick, dash dot},
    linepatmat/.style   = {smooth, myred, very thick},
    linetoppush/.style  = {smooth, cyan, very thick, dashed},
    linetoppushk/.style = {smooth, black, very thick, dotted},
    lineprimal/.style   = {smooth, myred, very thick},
    linedual/.style     = {cyan, very thick, dashed},
    legstyle/.style     = {column sep = 10pt, legend columns = -1}
}

\pgfplotstableread[col sep = comma]{game_2_DOvsFP_1.csv}\tabAConv
\pgfplotstableread[col sep = comma]{game_6_DOvsFP_1.csv}\tabBConv
\pgfplotstableread[col sep = comma]{game_6_DOvsFP_continuous_Alice_DO.csv}\tabBDens

\pgfplotstableread[col sep = comma]{Res_3_0.0625_bounds_equal_1.csv}\tabCaConv
\pgfplotstableread[col sep = comma]{Res_10_0.0625_bounds_unequal_1.csv}\tabCbConv

\usepgfplotslibrary{ternary}

\begin{document}
\maketitle

\begin{abstract}
  Many efficient algorithms have been designed to recover Nash equilibria of various classes of finite games. Special classes of continuous games with infinite strategy spaces, such as polynomial games, can be solved by semidefinite programming. In general, however, continuous games are not directly amenable to computational procedures. In this contribution, we develop an iterative strategy generation technique for finding a Nash equilibrium in a whole class of continuous two-person zero-sum games with compact strategy sets. The procedure, which is called the double oracle algorithm, has been successfully applied to large finite games in the past. We prove the convergence of the double oracle algorithm to a~Nash equilibrium. Moreover, the algorithm is guaranteed to recover an approximate equilibrium in finitely-many steps. Our numerical experiments show that it outperforms fictitious play on several examples of games appearing in the literature. In particular, we provide a detailed analysis of experiments with a version of the continuous Colonel Blotto game.
\end{abstract}

\noindent \section{Introduction}
Action spaces of games appearing in AI applications are often prohibitively large. Consequently, one has to strive for efficiently computable approximations of equilibria, possibly with provable bounds on convergence rates \cite{Gilpin12}. A number of algorithms applied in AI like fictitious play~\cite{Brown51}, the double oracle algorithm~\cite{McMahan03} or the~policy-space response oracle~\cite{LanctotZGLTPSG17,Muller19} overcome the problem with the cardinality by selecting `good' strategies iteratively. The selection process is usually based on an approximation of the best response. In~a~nutshell, the recent advances in algorithmic game theory has led to the  development of algorithms for (approximately) solving extremely large finite games, such as variants of poker \cite{moravvcik2017deepstack,BrownSandholm19} or multidimensional resource allocation problems \cite{Behnezhad17}. 

Completely new problems arise from considering games with infinite strategy spaces, in which the strategies are vectors of real numbers corresponding to physical parameters \cite{Archibald09} or to the setting of classifiers \cite{Loiseau19}. The first theoretical obstacle is that the existence of mixed strategy equilibria is guaranteed only for infinite games whose utility functions satisfy additional  conditions \cite{Glicksberg52,Fan52}. On top of that, some well understood classes of infinite games possess only optimal strategies whose supports are uncountable; see \cite{Roberson06} for an in-depth discussion of infinite Colonel Blotto games. 

Computational procedures for finding (approximate) equilibria of infinite games exist for rather special kinds of utility functions. Two-person zero-sum polynomial games are solvable by semidefinite programming; see \cite{ParriloIEEE06,LarakiLasserre12}. Approximate equilibria of separable games can be computed under additional assumptions \cite{SteinOzdaglarParrilo08}. However, games appearing in applications are rarely of the form above and a detailed analysis of their properties is inevitable; see \cite{Loiseau19} for an application in adversarial machine learning. Some authors develop approximations of best response by neural nets \cite{Kamra18,Kamra19}. One of the important iterative procedures for finite games, Brown-Robinson learning process known as fictitious play \cite{Brown51,Robinson51}, has been recently applied to infinite games  \cite{Ganzfried20}. However, the dynamics of best response strategies generated by fictitious play was analyzed only in special cases; cf. \cite{Hofbauer06,Perkins14}. To the best of our knowledge, not much is known about the convergence of fictitious play for general zero-sum continuous games as defined below. 

This paper deals with continuous games, which we define as two-person zero-sum games with continuous utility functions over compact strategy sets. We extend the double oracle algorithm \cite{McMahan03} to such games. This algorithm is an iterative strategy generation technique based on (i) the solution of subgames by LP solvers and (ii) the expansion of subgames' strategy sets using the best response strategies obtained thus far. Our main result is the convergence of this algorithm for any continuous game (Theorem \ref{thm:convergence}). The numerical experiments in Section \ref{sec:num} show that the double oracle algorithm converges faster than fictitious play on several examples (polynomial game, Townsend function, and a version of the Colonel Blotto game). The repository with our experiments' codes is \texttt{https://github.com/sadda/Double\_Oracle}.

\section{Basic Notions}
This section summarizes basic notions and results related to continuous zero-sum games and their equilibria; see \cite{Karlin59} or \cite{SteinOzdaglarParrilo08} for details.
\subsection{Continuous Games}
Player 1 and Player 2 select strategies from nonempty compact sets $X\subseteq \R^m$ and $Y\subseteq \R^n$, respectively. The utility function of Player $1$ is a continuous function $u\colon X\times Y\to \R$. The utility function of Player $2$ is $-u$. The triple $\calG=(X,Y,u)$ is called a \emph{continuous game.} Note that some authors use the term `continuous game' in a somewhat different sense allowing utility functions to be discontinuous functions over metric spaces of strategies.  

A continuous game $\calG=(X,Y,u)$ is (i) \emph{finite} if both $X$ and $Y$ are finite, and (ii) \emph{infinite} if $X$ or $Y$ is infinite. We will need the notion of subgame. When $X'\subseteq X$ and $Y'\subseteq Y$ are nonempty compact sets, we define the \emph{subgame} $\calG'=(X',Y',u)$ of $\calG$ by the restriction of $u$ to $X'\times Y'$, which is denoted by the same letter. 

The concept of mixed strategy in continuous games should allow every player to randomize with respect to any  probability measure on the corresponding strategy set. We will spell out the definitions related to mixed strategies only for Player 1. Their counterparts for Player 2 are completely analogous. A \emph{mixed strategy} of Player $1$ is a Borel probability measure $p$ over $X$. The set of all mixed strategies of Player $1$ is denoted by~$\Delta_X$. The \emph{support} of a mixed strategy $p\in\Delta_X$ is the set $$\spt p \coloneqq \bigcap \{K\subseteq X\mid K \text{ compact},\, p(K)=1\}.$$
Every mixed strategy $p\in \Delta_X$ can be classified as one of the following types depending on the size of its support.
\begin{enumerate}
    \item \emph{Pure strategy $p$}. This means that $\spt p=\{x\}$ for some $x\in X$. Equivalently, $p$ is equal to Dirac measure $\delta_x$.
    \item \emph{Finitely-supported mixed strategy $p$.} The support $\spt p$ is finite. Hence, $p$ can be written as a convex combination
    $$
    p = \sum_{x\in \spt p} p(x)\cdot \delta_x.
    $$
    \item \emph{Mixed strategy $p$ with infinite support $\spt p$.}  
\end{enumerate}

Put $\Delta\coloneqq \Delta_X\times \Delta_Y$. If players implement a mixed strategy profile $(p,q)\in \Delta$,
 the expected utility of Player $1$ is 
\begin{equation}\label{def:U}
  U(p,q) \coloneqq  \smallint_{X\times Y} u(x,y)\;\mathrm{d}(p\times q).
\end{equation}
This yields a function $U\colon \Delta\to \R$, which can be effectively evaluated in important special cases. For example, when both $\spt p$ and $\spt q$ are finite, 
$$
U(p,q) = \sum_{x\in \spt p}\; \sum_{y\in \spt q} p(x)\cdot q(x) \cdot u(x,y).
$$  
If Player $1$ employs a pure strategy given by $x\in X$ and Player 2 uses a mixed strategy $q\in\Delta_Y$, we will use the short notation $U(x,q) \coloneqq U(\delta_x,q)$.

\subsection{Equlibria in Continuous Games}
A mixed strategy profile $(p^*,q^*)\in\Delta$ is an \emph{equilibrium} in a~continuous game $\calG$ if 
\begin{equation}\label{NE}
  U(p,q^*) \leq U(p^*,q^*) \leq U(p^*,q)
\end{equation}
holds for all $(p,q)\in\Delta$. By Glicksberg's theorem \cite{Glicksberg52}, every continuous game has an equilibrium. Define the \emph{lower/upper value} of $\calG$ by
\begin{align*}
  \underline{v}(\calG)& \coloneqq \max\limits_{p\in\Delta_X} \min\limits_{q\in\Delta_Y} U(p,q) \quad \text{and}\\
  \overline{v}(\calG)& \coloneqq \min\limits_{q\in\Delta_Y} \max\limits_{p\in\Delta_X} U(p,q).
\end{align*}
Proposition \ref{prop:NE} gives several conditions for equilibrium, which will be used throughout the paper without further references. Its proof is omitted since it is completely analogous to the case of finite games.
\begin{proposition}\label{prop:NE}
  Let $\calG=(X,Y,u)$ be a continuous game and $(p^*,q^*)\in\Delta$. The following assertions are equivalent.
\begin{enumerate}
  \item The strategy profile $(p^*,q^*)$ is an equilibrium.
  \item $U(x,q^*) \leq U(p^*,q^*) \leq U(p^*,y)$ for all $(x,y)\in X\times Y$.
  \item $\min\limits_{y\in Y} U(p^*,y)=\underline{v}(\calG)$ and $\max\limits_{x\in X} U(x,q^*)=\overline{v}(\calG)$.
  \item $\underline{v}(\calG)= U(p^*,q^*) = \overline{v}(\calG)$.
\end{enumerate}
\end{proposition}
\noindent 
Hence, the equality $\underline{v}(\calG) = \overline{v}(\calG)$ holds for every continuous game $\calG$, and $v(\calG)\coloneqq \underline{v}(\calG)$ is called the \emph{value} of $\calG$.

Bounds on the size of supports of equilibrium strategies are known only for particular classes of continuous games, such as the class of separable games \cite{SteinOzdaglarParrilo08}. There are examples of games whose equilibria are almost any sets of finitely-supported mixed strategies \cite{Rehbeck18}. Moreover, some continuous games possess only equilibria with uncountable supports \cite{Roberson06}.

In many applications it is enough to find an \textit{$\epsilon$-equilibrium} $(p^*,q^*)$ for some $\epsilon\geq 0$, that is,
\begin{equation}\label{NE_eps}
  U(p,q^*)-\epsilon \leq U(p^*,q^*) \leq U(p^*,q)+\epsilon
\end{equation}
for all $(p,q)\in\Delta$. Note that this is a natural extension of~\eqref{NE}. According to Proposition \ref{prop:epsilonNE}, whose proof is in Appendix \ref{app:proofs}, we can always recover an approximate equilibrium $(p^*,q^*)$ with finite supports and such that $U(p^*,q^*)$ is arbitrarily close to the value of game $v(\mathcal G)$.

\begin{proposition}\label{prop:epsilonNE}
  Let $\calG$ be an arbitratry continuous game. Then for every $\epsilon>0$: 
  \begin{itemize}
    \item There exists an $\epsilon$-equilibrium $(p^*,q^*)$ of $\calG$ such that both $\spt p^*$ and $\spt q^*$ are finite.
    \item Every $\eps$-equilibrium $(p^*,q^*)$ of $\calG$ satisfies the inequality $|U(p^*,q^*) - v(\mathcal G)|\le\eps$.
  \end{itemize}
\end{proposition}

\section{Double Oracle Algorithm }

The double oracle algorithm uses the notion of best response strategies. For every mixed strategy $q\in \Delta_Y$ of Player 2, the~\emph{best response set} of Player 1 is
$$
\beta_1(q) \coloneqq \left\{x\in X\mid  U(x,q)=\max_{x'\in X} U(x',q)  \right\}.
$$
Analogously, for any $p\in \Delta_X$, put
$$
\beta_2(p) \coloneqq \left\{y\in Y\mid  U(p,y)=\min_{y'\in Y} U(p,y')  \right\}.
$$
Note that best response strategies are defined to be pure, without any loss of generality; see Proposition \ref{prop:switch}. Moreover, by compactness and continuity, $\beta_1(q)$ and $\beta_2(p)$ are always nonempty compact sets. 

The idea of the double oracle algorithm (Algorithm \ref{alg:do}) applied to a continuous game $\mathcal G=(X,Y,u)$ is simple. In every iteration, finite strategy sets $X_i$ and $Y_i$ are determined and some equilibrium $(p_i^*,q_i^*)$ of the finite subgame $(X_i,Y_i,u)$ is found by the standard linear programming methods. The best responses $x_{i+1}$ and $y_{i+1}$ to $q_i^*$ and $p_i^*$, respectively, are recovered, and added to the strategy sets. This is repeated until a terminating condition is satisfied. The resulting strategy profile is guaranteed to be an $\eps$-equilibrium.

\begin{algorithm}
  \caption{Double Oracle Algorithm}
  \label{alg:do}
  \begin{algorithmic}[1]
    \Require Continuous game $\calG=(X,Y,u)$, nonempty finite subsets $X_1\subseteq X$, $Y_1\subseteq Y$, and $\epsilon\geq 0$
    \State Let $i\coloneqq 0$
    \Repeat
    \State Increase $i$ by one
    \State Find an equilibrium $(p^*_i,q^*_i)$ of subgame $(X_i,Y_i,u)$
    \State Find some $x_{i+1}\in \beta_1(q^*_i)$ and $y_{i+1}\in \beta_2(p_i^*)$
    \State Let $X_{i+1}\coloneqq X_i \cup \{x_{i+1}\}$ and $Y_{i+1}\coloneqq Y_i \cup \{y_{i+1}\}$
    \State Let $\underline{v}_i\coloneqq U(p_i^*,y_{i+1})$ and $\overline{v}_i\coloneqq U(x_{i+1},q_i^*)$ 
    \Until{$\overline{v}_i - \underline{v}_i\leq \eps$}
  \Ensure $\eps$-equilibrium $(p_i^*,q_i^*)$ of game $\calG$
  \end{algorithmic}
\end{algorithm}

We now perform a simple analysis of the algorithm. Since
$$
U(p_i^*,q_i^*) = \max_{x \in X_i} U(x,q_i^*) \leq \max_{x \in X} U(x,q_i^*) = \overline{v}_i
$$
and similarly for the lower bound, we have
\begin{equation}\label{ineq:LU}
  \underline{v}_i\leq U(p_i^*,q_i^*) \leq \overline{v}_i.
\end{equation}
Lemma \ref{lemma:lower_upper} states that the same bounds hold even for the value of game $\calG$:
\begin{equation*}
  \underline{v}_i\leq v(\mathcal G) \leq \overline{v}_i.
\end{equation*}
The usual stopping condition of the double oracle algorithm for finite games is $X_{i+1}=X_i$ and $Y_{i+1}=Y_i$. Herein we chose the terminating condition $\overline{v}_i - \underline{v}_i\leq \eps$ for two reasons:
\begin{itemize}
  \item It is more general. Indeed, Lemma \ref{lemma:do_stop} states that if $X_{i+1}=X_i$ and $Y_{i+1}=Y_i$, then $\overline v_i - \underline v_i = 0$.
  \item It provides an estimate for the quality of approximate equilibrium. Formula \eqref{ineq:LU} implies $\overline v_i - \underline v_i \geq 0$ and Theorem \ref{thm:convergence} states that $(p_i^*,q_i^*)$ is an $(\overline v_i - \underline v_i)$-equilibrium. Then Proposition \ref{prop:epsilonNE} guarantees that $v(\calG)$ is known precisely up to $(\overline v_i - \underline v_i)$.
\end{itemize}

The main result of this manuscript is the convergence of the double oracle algorithm. In fact our result generalizes the result about convergence of the double oracle algorithm for finite games; see \cite{McMahan03}. For finite games, we neglect the case of $\eps>0$ since the algorithm is known to converge to an equilibrium for $\eps=0$ in finitely many steps. 
\begin{theorem}\label{thm:convergence}
Let $\mathcal G=(X,Y,u)$ be a continuous game.
\begin{enumerate}
  \item If $\calG$ is a finite game and $\eps=0$, Algorithm \ref{alg:do} converges to an equilibrium in a finite number of iterations.
  \item If $\calG$ is an infinite game and $\eps=0$, every weakly convergent subsequence of Algorithm \ref{alg:do} converges to an equilibrium in a possibly infinite number of iterations. Moreover, such a weakly convergent subsequence always exist.
  \item If $\calG$ is an infinite game and $\eps>0$, Algorithm \ref{alg:do} converges to a finitely supported $\eps$-equilibrium in a finite number of iterations.
\end{enumerate}
\end{theorem}
\begin{proof}
  We first realize that the terminating condition
  $$
  U(x_{i+1},q_i^*) - U(p_i^*,y_{i+1}) \le \eps 
  $$
  implies
  $$
  \aligned
  U(p_i^*,q_i^*) &\leq U(x_{i+1},q_i^*)\leq U(p_i^*,y_{i+1})  + \epsilon \\
  &= \min_{y'\in Y} U(p_i^*,y') +\epsilon = \min_{q\in\Delta_Y} U(p_i^*,q) + \epsilon.
  \endaligned
  $$
The first and the third relation above follow from the definition of the best response, the second from the terminating condition and the last from Proposition \ref{prop:switch}. Similarly, we can show that 
  $$
  \aligned
  U(p_i^*,q_i^*) &\geq U(p_i^*,y_{i+1}) \ge U(x_{i+1},q_i^*) - \epsilon \\
  &= \max_{x'\in X} U(x', q_i^*) -\epsilon = \max_{p\in\Delta_X} U(p,q_i^*) - \epsilon.
  \endaligned
  $$
  Combining these two inequalities implies that $(p_i^*,q_i^*)$ is an~$\eps$-equilibrium. Note that for $\eps=0$, this means that $(p_i^*,q_i^*)$ is an equilibrium.
  
  \emph{Item 1.} If $\calG$ is finite, then after a finite number of iterations it must happen that $X_{i+1}=X_i$ and $Y_{i+1}=Y_i$. Lemma \ref{lemma:do_stop} implies that the terminating condition of Algorithm \ref{alg:do} is satisfied with $\eps=0$ and the first paragraph of this proof implies that $(p_i^*,q_i^*)$ is an equilibrium of $\calG$.
  
  \emph{Item 2.} Consider now the case of an infinite game and $\eps=0$. If the double oracle algorithm terminates in a finite number of iterations, then the first paragraph implies that $(p_i^*,q_i^*)$ is an equilibrium. In the opposite case, the algorithm produces an infinite number of iterations. Due to Proposition~\ref{prop:convergence}, there is a weakly convergent subsequence which, for simplicity, will be denoted by the same indices. Therefore, $p_i^*\Rightarrow p^*$ for some $p^*$ and $q_i^*\Rightarrow q^*$ for some $q^*$, where the symbol $\Rightarrow$ denotes the weak convergence (Appendix \ref{app:weak}). 
  
  Consider any $y$ such that $y\in Y_{i_0}$ for some $i_0$. Take an~arbitrary $i\ge i_0$, which implies $y\in Y_i$. Since $(p_i^*,q_i^*)$ is an~equilibrium of the subgame $(X_i,Y_i,u)$, we get
  $$
    U(p_i^*,q_i^*) \le U(p_i^*,y) \to U(p^*,y),
  $$
  where the convergence follows from \eqref{eq:corv1}. Since $U(p_i^*,q_i^*)\to U(p^*,q^*)$ due to \eqref{eq:corv2}, this implies
  \begin{equation}\label{eq:conv_aux2}
    U(p^*,q^*)\le U(p^*,y)
  \end{equation}
  for all $y\in\cup Y_i$. Since $U$ is continuous, the previous inequality holds for all $y\in\cl(\cup Y_i)$.

  Fix now an arbitrary $y\in Y$. Because $y_{i+1}$ is the best response, we get
  \begin{equation}\label{eq:conv_aux3}
    U(p_i^*,y_{i+1}) \le U(p_i^*,y) \to U(p^*,y),
  \end{equation}
  where the limit holds due to \eqref{eq:corv1}. Since $y_{i+1}\in Y_{i+1}$ and by compactness of $Y$, we can select a convergent subsequence $y_i\to \hat y$, again without any relabelling, where $\hat y\in \cl(\cup Y_i)$. This allows us to use \eqref{eq:conv_aux2} to obtain
  \begin{equation}\label{eq:conv_aux4}
    U(p_i^*,y_{i+1}) \to U(p^*, \hat y) \ge U(p^*,q^*).
  \end{equation}
  Combining \eqref{eq:conv_aux3} and \eqref{eq:conv_aux4} yields
  $$
    U(p^*,q^*) \le U(p^*,y)
  $$
  for all $y\in Y$. Repeating the analogous arguments in the other variable yields
  $$
    U(x,q^*) \le U(p^*,q^*) \le U(p^*,y)
  $$
  for all $x\in X$ and $y\in Y$. Then Proposition \ref{prop:NE} says that $(p^*,q^*)$ is an equilibrium of $\calG$.
  
  \emph{Item 3.} Consider now the case of an infinite game with $\eps>0$ and realize that \eqref{eq:conv_aux3} and \eqref{eq:conv_aux4} also imply
  $$
    \aligned
    U(p^*,q^*)&\le U(p^*,\hat y) \leftarrow U(p_i^*,y_{i+1}) \\
    &\le U(p_i^*,q_i^*) \to U(p^*,q^*),
    \endaligned
  $$
  which means $U(p_i^*,y_{i+1})\to U(p^*,q^*)$. Similarly, $U(x_{i+1},q_i^*)\to U(p^*,q^*)$ and therefore
  $$
  \overline{v}_i - \underline{v}_i = U(x_{i+1},q_i^*) - U(p_i^*,y_{i+1}) \to 0.
  $$
  This states that the terminating condition will be satisfied after a finite number of iterations and the first paragraph of this proof states that $(p_i^*,q_i^*)$ is an $\eps$-equilibrium. Since only a finite number of iterations was performed and since $X_1$ and $X_2$ are finite, this implies that the supports of $p_i^*$ and $q_i^*$ are finite as well.
\end{proof}

Since best response strategies are not unique, in general, the sequence generated by Algorithm \ref{alg:do} may fail to converge for some continuous games. Hence, it is necessary to consider a convergent subsequence of iterates in Theorem~\ref{thm:convergence}. Such a continuous game is shown in Example \ref{example1}. Another feature of the double oracle algorithm is that the sequence $\overline v_i - \underline v_i$ has nonnegative terms and converges to~zero, but it is not necessarily monotone. This behavior can be demonstrated even for some finite games.

\section{Numerical Experiments}\label{sec:num}
We present two classes of games. The first class contains one-dimensional strategy spaces and the second class consists of certain Colonel Blotto games. The equilibrium of each finite subgame is found by solving a linear program. The best responses were computed by selecting the best point of a uniform discretization for the one-dimensional problems and by using a mixed-integer linear programming reformulation for the Colonel Blotto games. The examples were implemented in Python with solvers \texttt{scipy.optimize} and \texttt{mip}. All computations were performed on a laptop with Intel Core i5 CPU and 8GB RAM and no GPU was involved. Randomness is present only in the initialization of one-dimensional examples when a random pair of pure strategies is found.

We compare the double oracle algorithm with fictitious play. Its extension from finite to infinite games was recently formulated in \cite{Ganzfried20}.

\subsection{One-dimensional Examples}
We consider a polynomial game $\calG_1$ from \cite{ParriloIEEE06} with the strategy spaces $X=Y=[-1,1]$ and the utility function
$$
  u_1(x,y)=5xy-2x^2-2xy^2-y.
$$
In the equilibrium, Player 1 has the pure strategy $x^*=0.2$ and Player 2 has the mixed strategy $q^*=0.78\delta_{1} + 0.22\delta_{-1}$. The value of game is $-0.48$. Figure \ref{fig:game1_do_fp} shows the convergence of upper/lower estimates of the value of game. Note that the fictitious play is much slower to converge than the double oracle algorithm.

\begin{figure}[!ht]
  \centering
  \begin{tikzpicture}[scale=0.8]
    \begin{axis}[xlabel={Iteration},ylabel={Value of the game}, legend cell align={left}]
      \addplot [lineA] table[x index=0, y index=1] {\tabAConv};\addlegendentry{Double oracle}
      \addplot [lineB] table[x index=0, y index=3] {\tabAConv};  \addlegendentry{Fictitious play}
      \addplot [lineA] table[x index=0, y index=2] {\tabAConv}; \addplot [lineB] table[x index=0, y index=4] {\tabAConv};                 
    \end{axis}
  \end{tikzpicture}
  \caption{Convergence to the value of game $\calG_1$}
  \label{fig:game1_do_fp}
\end{figure}
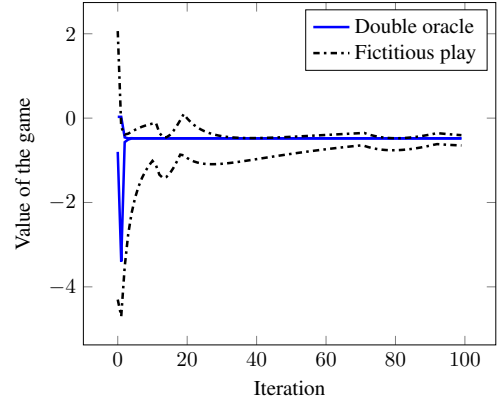

The utility function $u_2$ in our second example (game $\calG_2$) is based on \cite{Townsend}. Specifically, 
$$
  u_2(x,y) = -\cos^2((x-0.1)y) - x\sin(3x+y)
$$
is defined on $X = [-2.25, 2.5]$ and $Y = [-2.5, 1.75]$; see Figure \ref{fig:game6_f}. The convergence to the value is depicted on Figure~\ref{fig:game6_do_fp}. Once again the double oracle algorithm converges fast, while fictitious play is rather slow to converge. In Figure \ref{fig:game6_p} we show the optimal strategies of Player 1. The double oracle algorithm converged to a mixed strategy supported by~four points, the fictitious play seems to reach in limit a~continuous distribution whose peaks are those points. Note that the vertical axis is rescaled to account for the difference between discrete and continuous distributions.

\begin{figure}[!ht]
  \centering
  \includegraphics[width=0.8\linewidth]{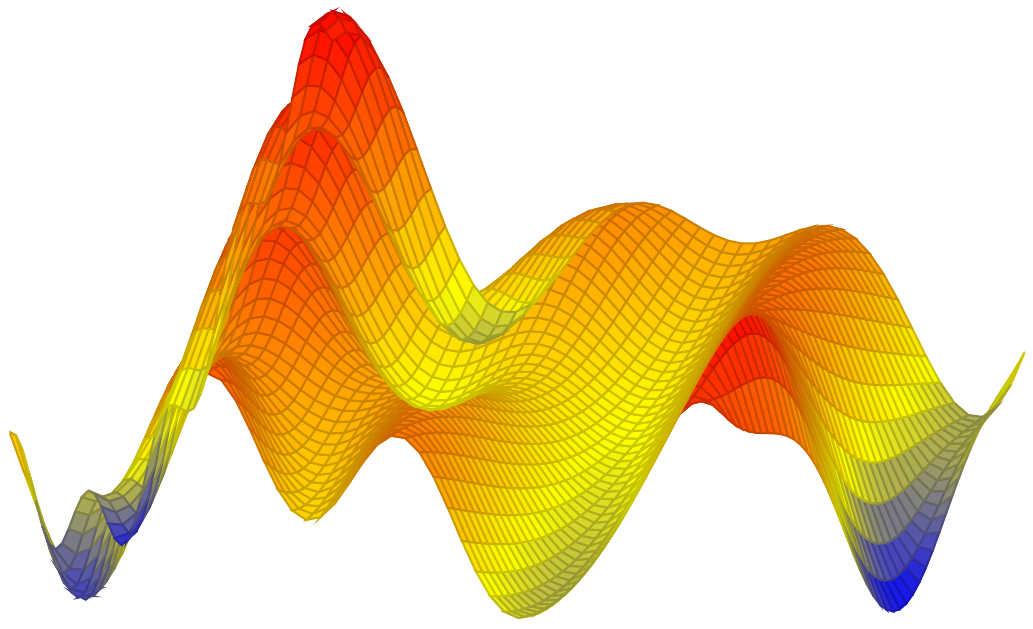}
  \caption{Townsend function $u_2$}
  \label{fig:game6_f}
\end{figure}

\begin{figure}[!ht]
  \centering
  \begin{tikzpicture}[scale=0.8]
    \begin{axis}[xlabel={Iteration},ylabel={Value of the game}, legend cell align={left}]
      \addplot [lineA] table[x index=0, y index=1] {\tabBConv};\addlegendentry{Double oracle}
      \addplot [lineB] table[x index=0, y index=3] {\tabBConv};  \addlegendentry{Fictitious play}
      \addplot [lineA] table[x index=0, y index=2] {\tabBConv}; \addplot [lineB] table[x index=0, y index=4] {\tabBConv};                 
    \end{axis}
  \end{tikzpicture}
  \caption{Convergence to the value of game $\calG_2$}
  \label{fig:game6_do_fp}
\end{figure}
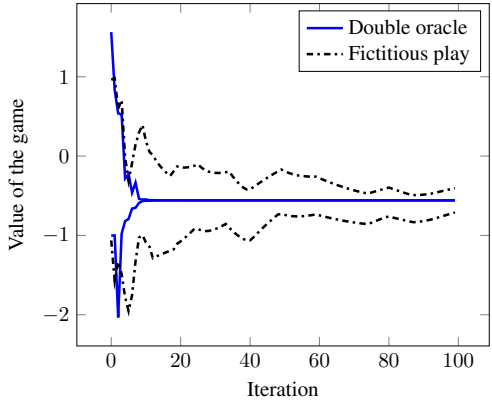

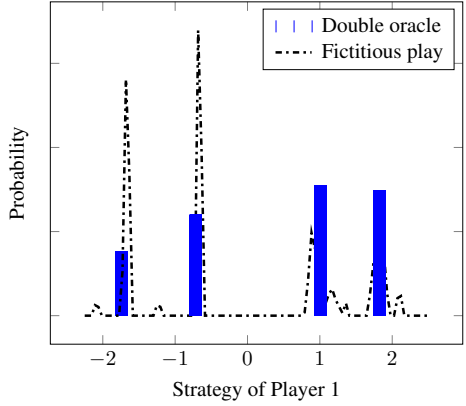
\begin{figure}[!ht]
  \centering
  \begin{tikzpicture}[scale=0.8]
    \begin{axis}[xlabel={Strategy of Player 1}, ylabel=Probability, yticklabel=\empty, legend cell align={left}]
      \addplot [lineC, ycomb] (-1.738,0.151840747);\addlegendentry{Double oracle}
      \addplot [lineB] table[x index=2, y index=3] {\tabBDens};  \addlegendentry{Fictitious play}
      \addplot [lineC, ycomb] (1.824,0.297440389);
      \addplot [lineC, ycomb] (-0.714,0.239661393);
      \addplot [lineC, ycomb] (1.005,0.310546245);
    \end{axis}      
  \end{tikzpicture}
  \caption{Mixed strategies in game $\calG_2$}
  \label{fig:game6_p}
\end{figure}

\subsection{Colonel Blotto Game}
We consider a continuous variant of the Colonel Blotto game. Two players simultaneously allocate forces across $n$ battlefields. Both strategy spaces $X$ and $Y$ equal to
$$
  \left\{\bm x \coloneqq (x^1,\dots,x^n)\in \R^n \mid x^j\ge0, \sum_{j=1}^n x^j=1\right\}.
$$
The utility function of Player $1$,
$$
  u(\bm x,\bm y) \coloneqq \sum_{j=1}^n a^j\cdot l(x^j - y^j),
$$
captures the total excess of the first army over the second army. The result on a battlefield $j$ is $a^j\cdot l(x^j - y^j)$, where $a^j>0$ is a weight of battlefield $j$ and $l(x^j - y^j)$ measures the 
performance of the first army on a battlefield $j$. The standard choice is the signum function $l(z)=\operatorname{sgn}(z)$; see \cite{GrossWagner50} or \cite{Roberson06}. In this paper we assume that each player must allocate a sufficiently higher proportion of forces than the opponent to win the battle on a single battlefield. Namely, we consider
\begin{equation}\label{eq:defin_l}
  l(z) = \begin{cases} -1 &\text{if }z\le -c, \\ \frac 1c z &\text{if }z\in[-c,c], \\ 1 &\text{if }z\ge c, \end{cases}
  \qquad \text{for some $c>0$.}
\end{equation}
When $c\to 0$, we recover the classical infinite colonel Blotto game since \eqref{eq:defin_l} approaches $\operatorname{sgn}(z)$ in the limit.

We will show how to compute best response strategies in case of~\eqref{eq:defin_l}. Assume that Player 2 employs strategies $(\bm y_1,\dots,\bm y_k)$ with  probabilities $(q_1,\dots,q_k)$, where $\bm y_i \coloneqq (y_i^1,\dots,y_i^n)\in Y$.  Then any best response strategy of~Player~$1$ is a solution~to
\begin{equation}\label{eq:opt_problem}
  \max_{\bm x \in X}\; \sum_{i=1}^k q_i \sum_{j=1}^n a^j\cdot l(x^j - y_i^j).  
\end{equation}
Since $l$ is a piecewise affine function, this nonlinear optimization problem can be reformulated as a mixed-integer linear problem. In Appendix \ref{app:br} we derive its equivalent form
$$
\aligned
\max_{\bm x,\bm s,\bm t,\bm z,\bm w}\quad &\sum_{i=1}^k q_i \sum_{j=1}^n a^j\left( s_{ij} - t_{ij} - 1\right) \\
\text{s.t.}\quad &\bm x\in X, \\
&s_{ij}\ge 0,\ s_{ij} \ge \tfrac 1c(x^j - y^j_i + c), \\
&s_{ij}\le \tfrac 1c(x^j - y^j_i + c) + M_l^s(1-z_{ij}),\\
&s_{ij}\le M_u^s z_{ij}, \\ 
&t_{ij}\ge 0,\ t_{ij} \ge \tfrac 1c(x^j - y^j_i - c), \\
&t_{ij}\le \tfrac 1c(x^j - y^j_i - c) + M_l^t(1-w_{ij}),\\
&t_{ij}\le M_u^t w_{ij}, \\
&s_{ij}\in\R,\ t_{ij}\in\R,\ z_{ij}\in\{0,1\},\ w_{ij}\in\{0,1\},
\endaligned
$$
where $
M_l^s = M_u^t= \tfrac 1c-1$ and $M_l^t = M_u^s = \tfrac 1c+1$. The~best response of Player~2 is obtained by solving an analogous MILP. Note that the MILP defined above is necessarily different from the one formulated in \cite{Ganzfried20}.

\begin{figure*}[!ht]
  \centering
  \includegraphics[scale=0.43]{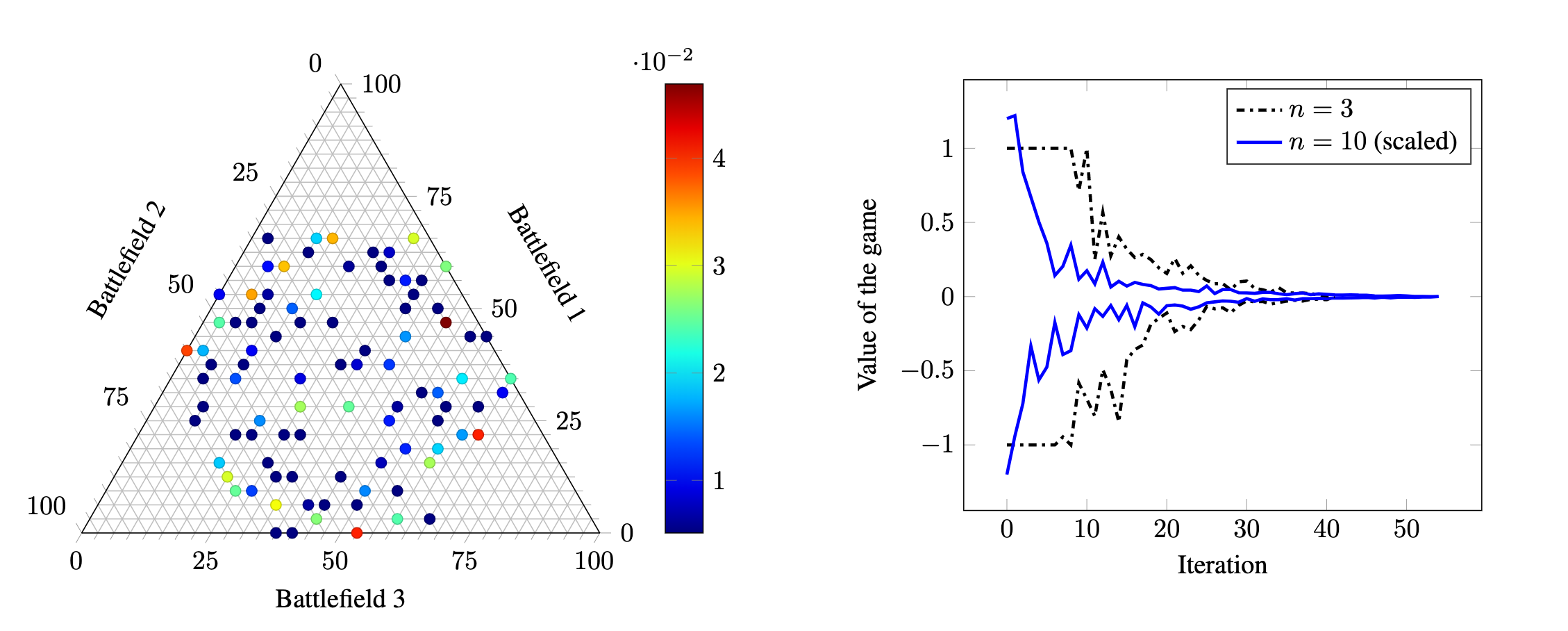}
  \caption{The optimal strategy for $c=\frac{1}{32}$ when started from three corner points (left). The convergence of the double oracle algorithm for $n=3$ and $n=10$ (scaled by $\frac{1}{50}$ for demonstration purposes) battlefields (right).}
  \label{fig:game_bounds}
\end{figure*}

\begin{figure*}[!ht]
  \centering
  \includegraphics[scale=0.4]{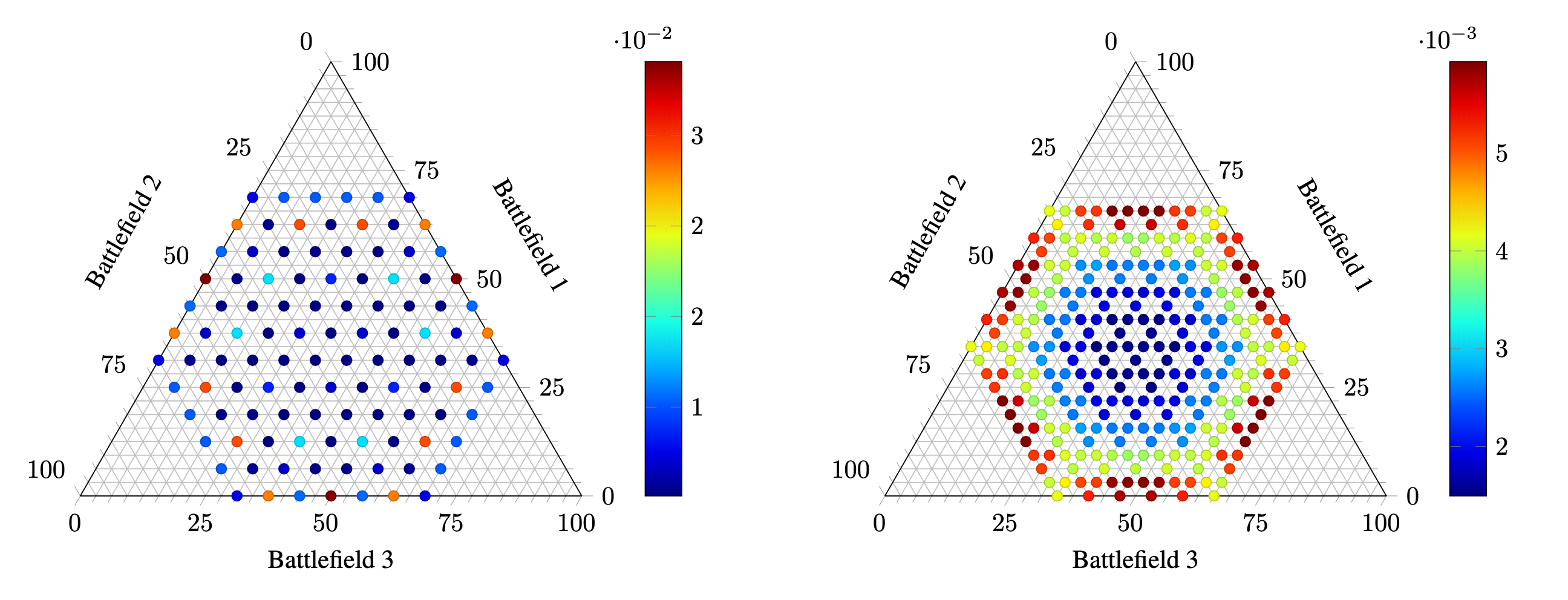}
  \caption{The optimal strategies for $c=\frac{1}{16}$ (left) and $c=\frac{1}{32}$ (right) produced by the double oracle algorithm when started from the grid. Both solutions are symmetric.}
  \label{fig:game_uniform}
\end{figure*}

For the numerical results we consider three battlefields ($n=3$) with equal weights ($\bm a=(1,1,1)$). We observed that the choice of initial strategy sets $X_1$ and $Y_1$ is crucial. Indeed, setting
$$
  X_1 = Y_1 = \{(1,0,0), (0,1,0), (0,0,1)\}
$$
provides much faster convergence than starting from a random point. The reason lies in the left-hand side of Figure~\ref{fig:game_bounds}, which shows the optimal solution produced by the double oracle algorithm for $c=\tfrac{1}{32}$. The optimal strategies are equidistant on the grid with distance $c$. This is a sensible result as the best response of Player 1 to the strategy $(y_1,y_2,y_3)$ of Player 2 is $(y_1+c,y_2+c,y_3-2c)$. Since $X_1$ and $Y_1$ already belong to the grid, all the iterates stay in it. However, they may not converge within this set when initial strategies are chosen at random.

The previous observation inspired us to start with both $X_1$ and $Y_1$ as the whole grid. It turned out that the double oracle converged in one iteration (the initial point was already an~equilibrium) to the strategies depicted in Figure \ref{fig:game_uniform}. The~left-hand side shows the results for $c=\frac{1}{16}$, while the right-hand side corresponds to $c=\frac{1}{32}$. These results are close to the hexagonal solutions obtained in \cite{GrossWagner50} and \cite{Roberson06}.

The right-hand side of Figure \ref{fig:game_bounds} shows the convergence of the double oracle algorithm for $n=3$ with $\bm a=(1,1,1)$ and for $n=10$ with $\bm a=(3,4,\dots,12)$. In both cases we put $c=\frac{1}{16}$. It appears  that the convergence is influenced by~$c$ more than by the number of battlefields $n$.

\section{Conclusions}
We extended the double oracle algorithm from finite to continuous games. We proved that the algorithm recovers a~finitely-supported $\eps$-equilibrium in finitely many iterations and converges to an equilibrium in a possibly infinite number of iterations. We showed that the double oracle algorithm performs better than fictitious play on selected examples. It is evident that the convergence of this algorithm depends on the size of constructed subgames and the best response calculation in each iteration. One of the open problems for future research is to analyze the speed of convergence of the~double oracle algorithm.

\section*{Acknowledgments}
This work was supported by the project RCI (CZ.02.1.01/0.0/0.0/16 019/0000765) ``Research Center for Informatics''. This material is based upon work supported by, or in part by, the Army Research Laboratory and the Army Research Office under grant number W911NF-20-1-0197.

\appendix

\section{Weak Convergence of Measures}\label{app:weak}
We will summarize a necessary background in weak topology on the space of probability measures \cite{Billingsley68}.
A sequence of mixed strategies $(p_i)$ in $\Delta_X$ \emph{weakly converges} to $p\in \Delta_X$ if
\begin{equation*}
  \lim_{i\to \infty} \smallint\nolimits_{X}f(x)\;\mathrm{d}p_i =  \smallint\nolimits_{X}f(x)\;\mathrm{d}p
\end{equation*}
for every continuous function $f\colon X\to \R$, and we denote this by $p_i\Rightarrow p$. Endowed with the topology corresponding to weak convergence, the convex set of mixed strategies $\Delta_X$ is a compact space. Analogously, $\Delta_Y$ becomes a compact set and so is the set $\Delta=\Delta_X\times \Delta_Y$. Then the definition~(\ref{def:U}) warrants that $U$ is a continuous function on $\Delta$. Note that compactness of $\Delta$ and continuity of $U$ imply the existence of all maximizers/minimizers throughout the paper.

\begin{proposition}\label{prop:convergence}
  The space $\Delta$ is weakly sequentially compact, that is, every sequence in $\Delta$ contains a weakly convergent subsequence.
\end{proposition}
\noindent  Since $U$ is continuous, the definition of weak convergence immediately implies the following two statements:
\begin{itemize}
  \item If $p_i\Rightarrow p$ in $\Delta_X$ and $q_i\Rightarrow q$ in $\Delta_Y$, then
  \begin{equation}\label{eq:corv2}
    U(p_i,q_i) \to U(p,q).
  \end{equation}
  \item If $p_i\Rightarrow p$ in $\Delta_X$ and $y_i\to y$ in $Y$, then
  \begin{equation}\label{eq:corv1}
    U(p_i,y_i) \to U(p,y).
  \end{equation}
\end{itemize}
Finally, it can be shown that the optimal value of utility function in response to the opponent's mixed strategy is attained for some pure strategy. 

\begin{proposition}\label{prop:switch}
For any $p\in\Delta_X$ we have
\begin{equation*}
  \aligned
  \min_{y\in Y} U(p,y) = \min_{q\in \Delta_Y} U(p,q).
  \endaligned
\end{equation*}
\end{proposition}

\section{Proofs and Additional Results}\label{app:proofs}

\begin{proof}[Proof of Proposition \ref{prop:epsilonNE}]
  The existence of an $\eps$-equilibrium follows from Theorem \ref{thm:convergence}.
  To prove the second part, assume that $(p^*,q^*)$ is an $\eps$-equilibrium. Then \eqref{NE_eps} implies
  \begin{equation}\label{eq:eps_equiv}
    \max_{p\in\Delta_X}U(p,q^*)-\epsilon \le U(p^*,q^*) \le \min_{q\in\Delta_Y} U(p^*,q)+\epsilon.
  \end{equation}
Let $(\hat p, \hat q)$ be an equilibrium of $\calG$. Then 
  $$
    U(\hat p,\hat q) \le U(\hat p, q^*) \le \max_{p\in\Delta_X} U(p,q^*)\le U(p^*,q^*)+\eps,
  $$
  where the first inequality follows from \eqref{NE} and the third from~\eqref{eq:eps_equiv}. In a similar way, we can show
  $$
    U(\hat p,\hat q) \ge U(p^*, \hat q) \ge \min_{q\in\Delta_Y} U(p^*,q)\ge U(p^*,q^*)-\eps.
  $$
  Combining these two relations with $U(\hat p,\hat q)=v(\mathcal G)$ imply the second statement of Proposition \ref{prop:epsilonNE}.
\end{proof}

\begin{lemma}\label{lemma:do_stop}
  Assume $X_{i+1}=X_i$ and $Y_{i+1}=Y_i$ in some step $i$ of Algorithm \ref{alg:do}. Then $U(p_i^*,y_{i+1})=U(x_{i+1},q_i^*)$.
\end{lemma}
\begin{proof}
  The condition $X_{i+1}=X_i$ implies $x_{i+1}\in X_i$. Then
  $$
  \aligned
  U(p_i^*,q_i^*) = \max_{x\in X_i}U(x,q_i^*) = \max_{x\in X}U(x,q_i^*) = U(x_{i+1},q_i^*),
  \endaligned
  $$
  where the first equality follows from Proposition \ref{prop:NE} applied to the subgame $(X_i,Y_i,u)$, the second from $x_{i+1}\in X_i$, and the third from the definition of iterate $x_{i+1}$.

  Similarly, we can show $U(p_i^*,q_i^*) = U(p_i^*,y_{i+1})$,
  which means $U(p_i^*,y_{i+1})=U(x_{i+1},q_i^*)$.
\end{proof}

\begin{lemma}\label{lemma:lower_upper}
  The inequality
  $$
  \underline v_i \le v(\mathcal G) \le \overline v_i
  $$
  holds in every step $i$ of Algorithm \ref{alg:do}.
\end{lemma}
\begin{proof}
  Let $(p^*,q^*)$ be an equilibrium of $\calG$. Then 
  $$
    \aligned
    \underline v_i &= U(p_i^*,y_{i+1}) = \min_{y\in Y} U(p_i^*, y) = \min_{q\in\Delta_Y} U(p_i^*, q) \\
    &\le U(p_i^*,q^*) \le U(p^*,q^*) = v(\mathcal G).
    \endaligned
  $$
  The second inequality can be obtained analogously.
\end{proof}

\begin{example}\label{example1}
  Define $X:=[0,1]$, $Y:=[0,1]$, and consider an arbitrary continuous function $u:X\times Y\to\R$ for which the double oracle algorithm produces an infinite number of iterates $(x_1,y_1),(x_2,y_2),\dots$ for $\eps=0$. Further, put $\tilde X:=[0,1]\cup [2,3]$ and let $\tilde u:\tilde X\times Y\to\R$ be given by
  $$
  \tilde u(x,y) = \begin{cases} u(x,y) &\text{if } x\in[0,1], \\ u(x-2,y) &\text{if } x\in[2,3]. \end{cases}
  $$
  Since $u$ is continuous, $(\tilde X, Y, \tilde u)$ is a continuous game. Since $\tilde u(x,y)=\tilde u(x+2,y)$, the extrema of marginal functions are not unique. Considering $\tilde y_i=y_i$, the double oracle algorithm may produce the sequence of iterations
  $$
  \tilde x_i = \begin{cases} x_i &\text{if } i\text{ is odd,}\\ x_i+2 &\text{if } i\text{ is even.}\end{cases}
  $$
  This sequence is obviously not convergent. However, there exists a convergent subsequence and its limit is an equilibrium by Theorem \ref{thm:convergence}.
\end{example}

\section{Best Response for Colonel Blotto Game}\label{app:br}

Function $l$ from \eqref{eq:defin_l} can be written as
$$
  l(z) = \max\left\{\tfrac 1c(z+c), 0\right\} - \max\left\{\tfrac 1c(z-c), 0\right\} - 1.
$$
With each $i,j$ in~\eqref{eq:opt_problem} we associate auxiliary variables $s_{ij}$ and~$t_{ij}$ and the contraints ensuring $l(x^j-y^j_i)=s_{ij}-t_{ij}-1$. The constraints on $s_{ij}$ and $t_{ij}$ follow from Lemma~\ref{lemma:mip}.

\begin{lemma}\label{lemma:mip}
  Let $a>0$, $b\in\R$, $M_l>0$, $M_u>0$ and $f(x) \coloneqq \max\{a(x-b), 0\}$. 
  For every $x$ such that $a(x-b) \in [-M_l,M_u]$ there are a unique $s\in\R$ and a possibly non-unique $z\in\{0,1\}$ solving the system
  $$
    \aligned
    s&\ge 0,  & s&\le a(x-b) + M_l(1-z),  \\
    s&\ge a(x-b), & s&\le M_uz.
    \endaligned
  $$
  Moreover, it holds $f(x)=s$.
\end{lemma}
\begin{proof}
  The proof is based on the well-known big-M method for the deactivation of constraints. The claim follows from the following implications,
  $$
  \aligned
  a(x-b) < 0 &\implies z=0 \implies s=0, \\
  a(x-b) > 0 &\implies z=1 \implies s=a(x-b).
  \endaligned
  $$
  If $a(x-b)=0$, then $s=0$ is unique, whereas $z$ may have either value.
\end{proof}
\noindent
Since $x^j,y^j_i\in[0,1]$, we have 
$$
  \aligned
  \tfrac 1c(x^j - y^j_i + c) \in [-\tfrac 1c+1, \tfrac1c+1], \\
  \tfrac 1c(x^j - y^j_i - c) \in [-\tfrac 1c-1, \tfrac1c-1],
  \endaligned
$$
which gives the bounds in Lemma~\ref{lemma:mip}.


\begin{thebibliography}{}
  \providecommand{\natexlab}[1]{#1}
  \providecommand{\url}[1]{\texttt{#1}}
  \providecommand{\urlprefix}{URL }
  \expandafter\ifx\csname urlstyle\endcsname\relax
    \providecommand{\doi}[1]{doi:\discretionary{}{}{}#1}\else
    \providecommand{\doi}{doi:\discretionary{}{}{}\begingroup
    \urlstyle{rm}\Url}\fi
  
  \bibitem[{Archibald and Shoham(2009)}]{Archibald09}
  Archibald, C.; and Shoham, Y. 2009.
  \newblock Modeling billiards games.
  \newblock In \emph{Proceedings of the 8th International Conference on
    Autonomous Agents and Multiagent Systems}, volume~1, 193--199.
  
  \bibitem[{Behnezhad et~al.(2017)Behnezhad, Dehghani, Derakhshan, HajiAghayi,
    and Seddighin}]{Behnezhad17}
  Behnezhad, S.; Dehghani, S.; Derakhshan, M.; HajiAghayi, M.; and Seddighin, S.
    2017.
  \newblock Faster and simpler algorithm for optimal strategies of {B}lotto game.
  \newblock In \emph{Proceedings of the 31st AAAI Conference on Artificial
    Intelligence}, 369--375.
  
  \bibitem[{Billingsley(1968)}]{Billingsley68}
  Billingsley, P. 1968.
  \newblock \emph{Convergence of probability measures}.
  \newblock New York: John Wiley \& Sons Inc.
  
  \bibitem[{Brown(1951)}]{Brown51}
  Brown, G.~W. 1951.
  \newblock Iterative solution of games by fictitious play.
  \newblock In \emph{Activity analysis of production and allocation}, 374--376.
    John Wiley \& Sons.
  
  \bibitem[{Brown and Sandholm(2019)}]{BrownSandholm19}
  Brown, N.; and Sandholm, T. 2019.
  \newblock Superhuman {AI} for multiplayer poker.
  \newblock \emph{Science} 365(6456): 885--890.
  
  \bibitem[{Fan(1952)}]{Fan52}
  Fan, K. 1952.
  \newblock Fixed-point and minimax theorems in locally convex topological linear
    spaces.
  \newblock \emph{Proceedings of the National Academy of Sciences of the United
    States of America} 38(2): 121.
  
  \bibitem[{Ganzfried(2020)}]{Ganzfried20}
  Ganzfried, S. 2020.
  \newblock Algorithm for Computing Approximate Nash equilibrium in Continuous
    Games with Application to Continuous Blotto.
  \newblock \emph{arXiv preprint arXiv:2006.07443} .
  
  \bibitem[{Gilpin, Pe\~{n}a, and Sandholm(2012)}]{Gilpin12}
  Gilpin, A.; Pe\~{n}a, J.; and Sandholm, T. 2012.
  \newblock First-order algorithm with $\mathcal{O}(\ln(1/\epsilon))$ convergence
    for $\epsilon$-equilibrium in two-person zero-sum games.
  \newblock \emph{Mathematical Programming} 133(1-2): 279--298.
  
  \bibitem[{Glicksberg(1952)}]{Glicksberg52}
  Glicksberg, I.~L. 1952.
  \newblock A further generalization of the {K}akutani fixed point theorem, with
    application to {N}ash equilibrium points.
  \newblock \emph{Proceedings of the American Mathematical Society} 3: 170--174.
  
  \bibitem[{Gross and Wagner(1950)}]{GrossWagner50}
  Gross, O.; and Wagner, R. 1950.
  \newblock A continuous {C}olonel {B}lotto game.
  \newblock Technical report, RM-408.
  
  \bibitem[{Hofbauer and Sorin(2006)}]{Hofbauer06}
  Hofbauer, J.; and Sorin, S. 2006.
  \newblock Best response dynamics for continuous zero-sum games.
  \newblock \emph{Discrete and Continuous Dynamical Systems--Series B} 6(1): 215.
  
  \bibitem[{Kamra et~al.(2018)Kamra, Gupta, Fang, Liu, and Tambe}]{Kamra18}
  Kamra, N.; Gupta, U.; Fang, F.; Liu, Y.; and Tambe, M. 2018.
  \newblock Policy learning for continuous space security games using neural
    networks.
  \newblock In \emph{Thirty-Second AAAI Conference on Artificial Intelligence},
    1103--1112.
  
  \bibitem[{Kamra et~al.(2019)Kamra, Gupta, Wang, Fang, Liu, and Tambe}]{Kamra19}
  Kamra, N.; Gupta, U.; Wang, K.; Fang, F.; Liu, Y.; and Tambe, M. 2019.
  \newblock DeepFP for Finding Nash Equilibrium in Continuous Action Spaces.
  \newblock In \emph{International Conference on Decision and Game Theory for
    Security}, 238--258. Springer.
  
  \bibitem[{Karlin(1959)}]{Karlin59}
  Karlin, S. 1959.
  \newblock \emph{Mathematical Methods and Theory in Games, Programming and
    Economics. Vol. 2: The Theory of Infinite Games}.
  \newblock Addison-Wesley Publishing Company.
  
  \bibitem[{Lanctot et~al.(2017)Lanctot, Zambaldi, Gruslys, Lazaridou, Tuyls,
    P{\'{e}}rolat, Silver, and Graepel}]{LanctotZGLTPSG17}
  Lanctot, M.; Zambaldi, V.~F.; Gruslys, A.; Lazaridou, A.; Tuyls, K.;
    P{\'{e}}rolat, J.; Silver, D.; and Graepel, T. 2017.
  \newblock A Unified Game-Theoretic Approach to Multiagent Reinforcement
    Learning.
  \newblock In \emph{Advances in Neural Information Processing Systems 30: Annual
    Conference on Neural Information Processing Systems}, 4190--4203.
  
  \bibitem[{Laraki and Lasserre(2012)}]{LarakiLasserre12}
  Laraki, R.; and Lasserre, J.~B. 2012.
  \newblock Semidefinite programming for min--max problems and games.
  \newblock \emph{Mathematical programming} 131(1-2): 305--332.
  
  \bibitem[{McMahan, Gordon, and Blum(2003)}]{McMahan03}
  McMahan, H.~B.; Gordon, G.~J.; and Blum, A. 2003.
  \newblock Planning in the presence of cost functions controlled by an
    adversary.
  \newblock In \emph{Proceedings of the 20th International Conference on Machine
    Learning (ICML-03)}, 536--543.
  
  \bibitem[{Morav{\v{c}}{\'\i}k et~al.(2017)Morav{\v{c}}{\'\i}k, Schmid, Burch,
    Lis{\'y}, Morrill, Bard, Davis, Waugh, Johanson, and
    Bowling}]{moravvcik2017deepstack}
  Morav{\v{c}}{\'\i}k, M.; Schmid, M.; Burch, N.; Lis{\'y}, V.; Morrill, D.;
    Bard, N.; Davis, T.; Waugh, K.; Johanson, M.; and Bowling, M. 2017.
  \newblock Deepstack: Expert-level artificial intelligence in heads-up no-limit
    poker.
  \newblock \emph{Science} 356(6337): 508--513.
  
  \bibitem[{Muller et~al.(2019)Muller, Omidshafiei, Rowland, Tuyls, Perolat, Liu,
    Hennes, Marris, Lanctot, Hughes et~al.}]{Muller19}
  Muller, P.; Omidshafiei, S.; Rowland, M.; Tuyls, K.; Perolat, J.; Liu, S.;
    Hennes, D.; Marris, L.; Lanctot, M.; Hughes, E.; et~al. 2019.
  \newblock A generalized training approach for multiagent learning.
  \newblock \emph{arXiv preprint arXiv:1909.12823} .
  
  \bibitem[{Parrilo(2006)}]{ParriloIEEE06}
  Parrilo, P. 2006.
  \newblock Polynomial games and sum of squares optimization.
  \newblock In \emph{Decision and Control, 2006 45th IEEE Conference on},
    2855--2860.
  
  \bibitem[{Perkins and Leslie(2014)}]{Perkins14}
  Perkins, S.; and Leslie, D.~S. 2014.
  \newblock Stochastic fictitious play with continuous action sets.
  \newblock \emph{Journal of Economic Theory} 152: 179--213.
  
  \bibitem[{Rehbeck(2018)}]{Rehbeck18}
  Rehbeck, J. 2018.
  \newblock Note on unique {N}ash equilibrium in continuous games.
  \newblock \emph{Games and Economic Behavior} 110: 216--225.
  
  \bibitem[{Roberson(2006)}]{Roberson06}
  Roberson, B. 2006.
  \newblock The {C}olonel {B}lotto game.
  \newblock \emph{Economic Theory} 29: 1--24.
  
  \bibitem[{Robinson(1951)}]{Robinson51}
  Robinson, J. 1951.
  \newblock An iterative method of solving a game.
  \newblock \emph{Annals of Mathematics} (54): 296--301.
  
  \bibitem[{Stein, Ozdaglar, and Parrilo(2008)}]{SteinOzdaglarParrilo08}
  Stein, N.~D.; Ozdaglar, A.; and Parrilo, P.~A. 2008.
  \newblock Separable and low-rank continuous games.
  \newblock \emph{International Journal of Game Theory} 37(4): 475--504.
  
  \bibitem[{Townsend(2014)}]{Townsend}
  Townsend, A. 2014.
  \newblock Constrained optimization in {C}hebfun.
  \newblock
    \urlprefix\url{http://www.chebfun.org/examples/opt/ConstrainedOptimization.html}.
  
  \bibitem[{Yasodharan and Loiseau(2019)}]{Loiseau19}
  Yasodharan, S.; and Loiseau, P. 2019.
  \newblock Nonzero-sum Adversarial Hypothesis Testing Games.
  \newblock In \emph{Advances in Neural Information Processing Systems},
    7310--7320.
  
  \end{thebibliography}
\end{document}